\pdfoutput=1
\documentclass{article}%
\usepackage{amssymb}
\usepackage{amsfonts}
\usepackage{setspace}
\usepackage{amsmath}
\usepackage{graphicx}
\usepackage{tikz}%
\providecommand{\U}[1]{\protect\rule{.1in}{.1in}}
\newtheorem{theorem}{Theorem}

\newtheorem{example}[theorem]{Example}

\newtheorem{lemma}[theorem]{Lemma}

\newtheorem{proposition}[theorem]{Proposition}
\newtheorem{remark}[theorem]{Remark}

\newenvironment{proof}[1][Proof]{\noindent\textbf{#1.} }{\ \rule{0.5em}{0.5em}}
\begin{document}

\begin{center}
{\LARGE Some new results on Duffie-type OTC markets}

\bigskip

BY ALAIN$\ $B\'{E}LANGER$^{\ast}${\large ,\ }GASTON GIROUX$^{\ast}%
${\small \ \bigskip}AND NDOUN\'{E}\ NDOUN\'{E}$^{\ast}$

\textit{Universit\'{e} de Sherbrooke}
\end{center}

\bigskip

Abstract: {\small The extended Wild sums considered in this article generalize
the classical Wild sums of statistical physics. We first show
how to obtain explicit solutions for the evolution equation of a large system
where the interactions are given by a single, but general, interacting kernel
which involves }${\small m}$ {\small components, for a fixed }$m\geq
2.${\small We then show how to retain the explicit formulas for the case of
OTC market models where the dynamics is more directly described by two (or
more) kernels. }

\section{Introduction}

After the publication of  M. Kac's work (1956) [8], there was a renewed interest for the results of E. Wild (1951) [15]. This interest was mainly focused on the random matching of a large population of particles forming a diluted Maxwell gas. Here we develop an approach inspired by this body of work. To do so, we start with a
sequence of dynamical sets of interacting components, one for each integer
$N.$ For these dynamical systems we can show that when $N$ is large the
probability is very small that a component has interacted more than once,
directly or indirectly, up to time $t,$ with any other component. Thanks to
this fundamental property, we can link the microscopic and macroscopic levels
using results from the theory of continuous-time Markov chains.

The Wild sum is a series construction which gives the solution of a given
evolution equation in the statistical physics of gases as first appeared in
the work of E. Wild [15]. Note that the classical expression of a Wild sum is described  by binary
trees. Inspired by these ideas, S. Tanaka [12] and H. Tanaka [11] defined an
extension of Wild's sum for solving certain non-linear differential equations
of spaces of measures, so the expression of this sum is described by appropriate
trees. However, the problem of showing the existence of these sums remains
wide open in general.

The recursive time relaxed Monte Carlo methods of Trazzi, Pareschi and
Wennberg [14] are based on generalized Wild sums. However, the lack of
explicit formulas for these sums constitutes a handicap for the efficiency of the above

\bigskip

\noindent$\overline{%
\begin{array}
[c]{l}%
\text{{\small Dated 7 March 2015.}}\\
^{\ast}:\text{{\small D\'{e}partement de finance, Universit\'{e} de
Sherbrooke,}}%
\end{array}
}$\\
\noindent{\small Sherbrooke, Canada, J1K 2R1. E-mails:
alain.a.belanger@usherbrooke.ca,\\ 
gasgiroux@hotmail.com, ndoune.ndoune@usherbrooke.ca}\\
{\small \noindent}\noindent{\small AMS classifications:\ 60G55, 34A34,
82C31.\\}
{\small \noindent}\noindent{\small Keywords: Large interacting sets, Ordinary
Differential Equations, continuous-time Markov chains.}

\newpage \noindent{ methods as well as others also based on extended Wild sums
(see [13], for instance).}

Carlen et al [4] obtain Wild sum formulas which are quite explicit for the solution of
the Kac equation. Their binary trees are obtained, in the spirit of McKean,
from commutator formulas for Lie algebras, leading them to groupings of
interaction trees. Consequently, our more general interaction trees are
different form theirs even in the binary case.

The aim of this paper is to propose a combinatorial formula for extended Wild
sums which are solutions of certain evolution equations and more precisely in
the context of interactions involving $m$ components, $m\geq2$.

In section 3 of B\'{e}langer-Giroux [1], the explicit formulas for the Wild sums were used to obtain the convergence of the solution of the evolution equation to a steady state. This is one of the important applications permitted by the tractability of our explicit formulas.

The article is organized as follows. In section 2, we introduce the types of
combinatorial trees which are going to be useful in the expression of the
solution of the evolution equation in terms of interaction trees. In section 3
and 4 we consider interactions involving $m$ components, for $m\geq2,$ and we
suppose that the intensities of these dynamics have an adequate dependence on
$N$. Our techniques enable us to obtain an explicit formula for the solution of the associated
system of differential equations. In section 4, we show how to retain the
explicit formulation of the solutions in the case of OTC market models
described by two kernels.

\section{\bigskip Combinatorial trees}

We assume that the reader is familiar with the basic definitions of trees. A
rooted tree is a tree with a designated node called the root. A rooted tree in
which the rooted node has one child is a planted tree. An $m$-ary tree is a
rooted tree where each of its node is either a leaf (that is, it has no child)
or it has exactly $m$ children. The leafs are called external nodes and those
nodes with $m$ children, internal nodes. Note that we do not consider the root of the tree as an internal node.\newline An ordered tree is a rooted
tree in which the children of each node are assigned a fixed ordering.\newline
A rooted tree is called an $(m,1)$-ary tree if each internal node has either
one child or exactly $m$ children. In this article, we will work with ordered
$m$-ary trees and ordered $(m,1)$-ary trees.

Let $\mathbb{A}_{n}$ denote the set of $m$-ary ordered 
trees with $n$ internal nodes. Each tree in $\mathbb{A}_{n}$ has $(m-1)n+1$ leaves and each tree can be obtained by adding an internal node on a leaf of a tree in $\mathbb{A}_{n-1}$ (taking into account the order). Hence the number of trees in $\mathbb{A}_{n}$ is
 $\#_{m}(n)$ $=%
 {\displaystyle\prod\limits_{k=1}^{n-1}}
 ((m-1)k+1)$. 

\section{The dynamics}

Let $N$ be the (large) number of interacting components. Let $m$ ($m\geq2$) be
the fixed number of components involved in each interaction. We suppose that
all components take their values in a measurable space, $(E,\mathcal{E}),$
(one can think of ($%
\mathbb{R}
^{d},B(%
\mathbb{R}
^{d})$ or simply a finite set$)$ and their interactions are given by a
symmetric probability kernel $Q$ on the product space ($E^{m},\mathcal{E}%
^{\otimes m}$ ) $.$ That is, the function $Q(x_{1},x_{2},...x_{m};C_{1}%
\times\cdot\cdot\cdot\times C_{m})$: is measurable in $(x_{1},x_{2},...x_{m});$
is a probability measure in ($C_{1}\times\cdot\cdot\cdot\times C_{m});$ and
satisfies $Q(x_{1},x_{2},...x_{m};C_{1}\times\cdot\cdot\cdot\times
C_{m})=Q(x_{\sigma(1)},x_{\sigma(2)},...x_{\sigma(m)};C_{\sigma(1)}\times
C_{\sigma(2)}\times\cdot\cdot\cdot\times C_{\sigma(m)})$ for any permutation
$\sigma$ of $\{1,2,...,m\}.$ 

In the following example, we simplify the model of Duffie-G\^{a}rleanu-Pedersen [6] by keeping only their binary interacting kernel.

\begin{example}
.\ Investors in this model have two liquidity
states denoted $h,$ for high, and $l$, for low. Moreover,there is an asset of common interest to these investors who either own
the asset (denoted by $o)$ or don't (denoted by $n$).\ So $E=\mathbb{\{(}l,n),(l,o),(h,n),(h,0)\}$ describes the state space. The kernel
is defined by $Q_{2}(\cdot,\cdot;C_{1}\times C_{2})=0$ except for
\[
Q_{2}((h,n),(l,o);C_{1}\times C_{2})=Q_{2}((l,o),(h,n);C_{2}\times
C_{1})=\delta_{(h,o)}(C_{1})\delta_{(l,n)}(C_{2})
\]
where $\delta_{z_{0}}$ is the Dirac function $\delta_{z_{0}}(z)=1$ iff
$z=z_{0}$ and $\delta_{z_{0}}(z)=0$ otherwise. The binary kernel implements
the trading of the asset whenever a low liquidity investor who owns the asset
meets a high liquidity investor who does not yet hold it.
\end{example}

The interactions occur at each jump of a Poisson process with intensity
$\lambda\frac{N}{m}$. \ Groups are undistinguishable so each group has a
probability of $\binom{N}{m}^{-1}$ of being involved in a given interaction.

The kernel $Q$ allows us to describe the macroscopic evolution of the system
with an associated system of non-linear differential equations via the
evolution of the law of a component. This probability law, denoted $\mu_{t}, $
evolves with time and is in fact the solution of the Cauchy problem:%
\[
\frac{d\mu_{t}}{dt}=\lambda (\mu_{t}^{\circ_{m}}-\mu_{t})\;;\mu_{0}=\mu
\]

where

\bigskip%
\[
\mu^{\circ_{m}}(C)\triangleq%
{\displaystyle\int\limits_{\mathbb{R} ^{m}}}
\mu(dx_{1})\mu(dx_{2})...\mu(dx_{m})Q(x_{1},x_{2},...x_{m};C\times
E^{m-1})\text{ for }C\in\mathcal{E}.
\]

$\ $ The probability law $\mu^{\circ_{m}}$is the law of a component after the
interaction of $m$ i.i.d. components with law $\mu.$ We can think of it as the
law at the root of the $m$-ary tree with only one interaction. We will look at
all the trees representing the interaction history of a component up to time
$t$. So for a tree, $A$, with more than one interaction, we divide the tree in
$m$ subtrees at that last interaction and continue recursively up to time 0 to
define $\mu^{\circ_{m}A}$ . (Please see figure 1 for a simple example of an
interaction tree.) Let $\mathbb{A}_{n}$ be the set of all trees with $n$
interactions (a.k.a. nodes), each node producing $m$ branches. If $A_{n}%
\in\mathbb{A}_{n}$, then $\mu^{\circ_{m}A_{n}}$ denotes the law obtained by
iteration of $\mu^{\circ_{m}}$ through the successive nodes of the tree when
we place the law $\mu$ on each leaf of $A_{n}.$%

\begin{figure}
[ptb]
\begin{center}
\includegraphics[
natheight=4.249700in,
natwidth=5.416300in,
height=4.3016in,
width=5.4743in
]%
{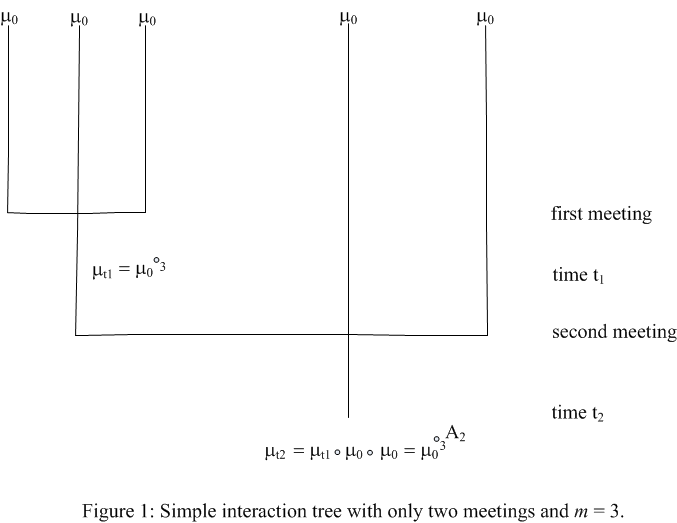}%
\end{center}
\end{figure}

\bigskip

\bigskip

\bigskip

We have shown B\'{e}langer-Giroux [1] that the Cauchy problem has a unique
solution which can be expressed, by conditioning on the number of interactions
up to time $t,$ and then by the component's history. Such conditionings give us%

\[
\mu_{t}=%
{\displaystyle\sum\limits_{n\geq0}}
p_{n}(t)\frac{1}{\#_{m}(n)}%
{\displaystyle\sum\limits_{A_{n}\in\mathbb{A}_{n}}}
\mu^{\circ_{m}A_{n}}\text{ \ \ \ \ \ \ }(1)
\]

where $\#_{m}(n)$ $=%
{\displaystyle\prod\limits_{k=1}^{n-1}}
((m-1)k+1)$ is the number of trees with $n$ nodes, taking into account their
branching orders; and $p_{n}(t)$ $=\frac{\#_{m}(n)}{(m-1)^{n}n!}%
e^{-\lambda t}(1-e^{-(m-1)\lambda t})^{n}$ is the probability of having $n$ branchings up to
time $t.$

\begin{remark}
We call the law $\mu_{t}=%
{\displaystyle\sum\limits_{n\geq0}}
e^{-\lambda t}(1-e^{-(m-1)\lambda t})^{n}\frac{1}{(m-1)^{n}n!}%
{\displaystyle\sum\limits_{A_{n}\in\mathbb{A}_{n}}}
\mu^{\circ_{m}A_{n}}$ an explicit extended Wild sum [15] and note that the
convex combination we obtain for the case $m=2$ is indeed the Wild sum,
$\mu_{t}=%
{\displaystyle\sum\limits_{n\geq0}}
e^{-\lambda t}(1-e^{-\lambda t})^{n}\frac{1}{n!}%
{\displaystyle\sum\limits_{A_{n}\in\mathbb{A}_{n}}}
\mu^{\circ_{m}A_{n}}$, now well-known in the statistical physics of gases
since the work of Kac (1956) [8].
\end{remark}

\subsection{\bigskip Using interaction trees to go from the microscopic to the
macroscopic.}

In all our cases, we have an underlying market structure which is a Kac walk
with interactions involving $m$ agents$.$ We add exponential times to obtain a
marked Poisson process whose marks are horizontal lines linking the agents
participating in a given interaction. This enabled us, in B\'elanger-Giroux [1], to describe the limit
law of an agent, under an appropriate conditioning, as a countable convex
combination on trees which is, as we have shown in section 3 of that article, the global
solution of the associated differential equation on the space of probability laws.

Here we first explain how we came to that convex combination since it serves as a tool to study the other models which follow. It is the tool that enables us, for instance, to state proposition 5. Its proof follows the lines of the proof of the main result in B\'elanger-Giroux [1].

We start our study by
an analysis of the dynamics of the intrinsic structure of the large set of
interacting agents when the number of agents increases. We assume that each
interaction involves $m$ agents, $m\geq2$. More specifically, we consider a
set of $N$ agents whose interactions happen at unexpected times so these
interactions' occurrences follow a Poisson process$.$ Since agents are
interchangeable, each group has an equal probability of meeting of $\left(
\begin{array}
[c]{c}%
N\\
m
\end{array}
\right)  ^{-1}.$ If we suppose the intensity of the meetings to be $\frac
{N}{m}$ then each agent has a meeting rate $\lambda$ which can be assumed to
equal $1$ under a time change. We will make this assumption, $\lambda=1$, all throughout section 3.

For $N$ fixed and starting at time $0$, we assign a vertical position to each
agent. The down movement represents the passage of time, see figure 1 on page
4. Each time a group of agents interacts, we draw a horizontal line between
those agents and we draw a vertical line at each agent's position connecting
$0$ to the horizontal line just drawn, so we see a random graph being formed.
When we stop this graph at time $t$, we obtain the finite graph of all
interactions that have taken place. Moreover, the history up to time $t$ of a
given agent, call it $P$, is described by the random graph connecting all
agents who have interacted directly or indirectly with $P$. 

\begin{figure}
[ptb]
\begin{center}
\includegraphics[
natheight=4.249700in,
natwidth=5.416300in,
height=4.3016in,
width=5.4743in
]%
{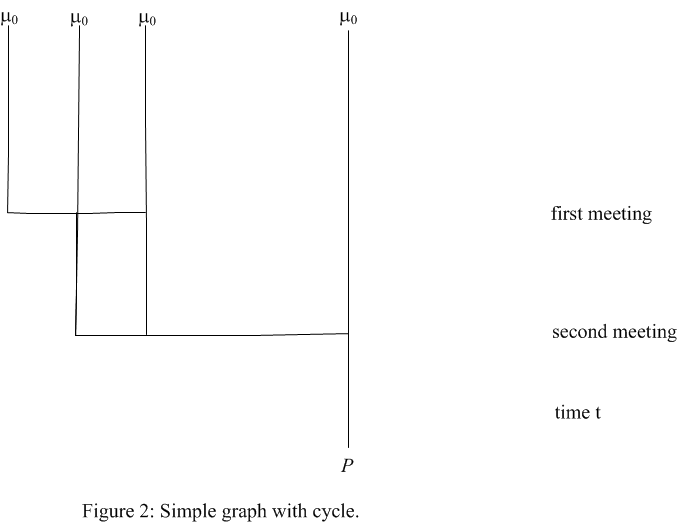}%
\end{center}
\end{figure}

The number of meetings is random but we can condition on it. The law of the
finite graph is reversible since the meeting times are uniform on $[0,t]$. We
want to show that a random graph representing the history of $P$ can be
replaced by a random tree as the number of agents, $N$, grows. If we look at
figure 2, we see that the inclusion in the second meeting of one of the
investors having participated in the first one would create a cycle in
our graph. As $N$ grows though, the chance of meeting an investor
previously encountered directly or indirectly tends to zero.

To see this, let us consider the graph of $P$'s history up to time $t$.
Starting at time $t,$ \ we pursue each one of the encountered vertical lines
in $P$'s history backward in time until we reach the next horizontal line. If
the inclusion of the horizontal line in our graph does not create a cycle
(i.e. no pair of investors were involved directly or indirectly in a previous
meeting) we include the line, if not we remove it. Proceeding in this fashion
up to time $0$ we get a tree with $n$ internal nodes, say, which has the same law
as the law of a tree obtained by a pure-birth process. The tree obtained by a
sample history of $P$'s interactions is an $m$-ary tree.

These trees grow randomly in time: each time a new node appears, corresponding
to the occurrence of a meeting of investors at that time. We recall that $\mathbb{A}_{n}$ denotes the set of $m$-ary ordered trees with $n$ internal nodes. Then $\mathbb{A}_{n}$ constitutes a set of random trees if we
assume that every $m$- ary tree in $\mathbb{A}_{n}$ is equally likely, namely
of probability $\frac{1}{\#_{m}(n)}.$

The tree starting at $P^{\prime}$s vertical line at time $t$ with
intensity $1$ and which at time $0$ has intensity $(m-1)n+1$ and that same
number of leaves. Between two branchings of this process a graph representing
$P$'s meeting history can have a random number of additional horizontal lines
following a Poisson law of parameter at most $\frac{N}{m}\left(  \left(
\begin{array}
[c]{c}%
(m-1)n+1\\
2
\end{array}
\right)  \left(
\begin{array}
[c]{c}%
N\\
m
\end{array}
\right)  ^{-1}\right)  $. We will now bound the expectation of these
supplementary horizontal lines by a majorant which tends to $0$ as $N$
increases. Indeed, since the mean number of redundant lines when there are $n
$ branchings up to time $t$ is at most $\frac{N}{m}\left(  \left(
\begin{array}
[c]{c}%
(m-1)n+1\\
2
\end{array}
\right)  \left(
\begin{array}
[c]{c}%
N\\
m
\end{array}
\right)  ^{-1}\right)  ,$ we have that the mean number of redundant horizontal
lines is bounded above by
\[%
{\displaystyle\sum\limits_{n\geq0}}
\frac{N}{m}\left(  \left(
\begin{array}
[c]{c}%
(m-1)n+1\\
2
\end{array}
\right)  \left(
\begin{array}
[c]{c}%
N\\
m
\end{array}
\right)  ^{-1}\right)  p_{N,n}(t),
\]
where $p_{N,n}(t)$ is the probability of having $n$ branchings up to time $t$
of the pure birth process with successive branching waiting times following
exponential laws of parameter
\[
\lambda_{N,n}=\frac{N}{m}(\left(  m-1)n+1\right)  \left(
\begin{array}
[c]{c}%
N-((m-1)n+1)\\
m-1
\end{array}
\right)  \left(
\begin{array}
[c]{c}%
N\\
m
\end{array}
\right)  ^{-1}
\]

\bigskip Since
\begin{align*}
&  \frac{N}{m}(\left(  m-1)n+1\right)  \left(
\begin{array}
[c]{c}%
N-((m-1)n+1)\\
m-1
\end{array}
\right)  \left(
\begin{array}
[c]{c}%
N\\
m
\end{array}
\right)  ^{-1}\\
&  =\frac{((m-1)n+1)\binom{N-((m-1)n+1)}{m-1}}{\binom{N-1}{m-1}}\text{
\ \ \ \ \ \ \ \ \ \ \ }(2)\\
&  \leq(m-1)n+1
\end{align*}
then $p_{N,n}(t)$ is stochastically smaller than the law obtained with the
intensities $\lambda_{n}=(m-1)n+1,$ which in turn are less than the
intensities $\overline{\lambda}_{n}=m(n+1).$ Its transition kernel is then
obtained by solving Kolmogorov's affine system of equations:%

\begin{align*}
\frac{d\overline{p}_{t}(0)}{dt}  &  =-m\overline{p}_{t}(0)\text{
\ \ \ \ \ \ \ \ \ \ \ \ \ \ \ \ \ \ \ \ \ \ \ \ \ \ \ \ \ \ \ \ \ \ \ \ }%
\ \text{\ \ \ \ \ \ \ \ \ \ \ \ \ \ \ }\\
\frac{d\overline{p}_{t}(n)}{dt}  &  =mn\overline{p}_{t}(n-1)-m(n+1)\overline
{p}_{t}(n)\text{ \ ; }n\geq1\text{.}%
\end{align*}
Thus the latter intensities give us a geometric law $\overline{p}%
_{t}(n)=e^{-mt}(1-e^{-mt})^{n}\ 
=e^{-m(n+1)t}(e^{mt}-1)^{n}$. Since geometric laws have finite moments of
all orders, the mean number of redundant horizontal lines is bounded above by
a quantity converging to $0$.

For more details on Kolmogorov systems of equations for pure birth processes
we refer the reader to Lefebvre [10], for instance.

Thus, after having specified the initial agents' states and their interaction
kernels, we can approximate $P^{\prime}s$ law using the tree obtained from
removing all redundant horizontal lines from its graph. We will use this fact
in the next sub-section.

\subsection{Limit countable convex combination}

\bigskip We will now show that these random trees whose branching intensities
depend on $N$ can be approximated by trees with branching intensities
independent of $N$. Taking into account that $P$'s tree history is random with
intensities depending on $N,$ we could write $P$'s law, denoted by $\mu
_{t}^{\ast,N},$ with complex formulae depending on $N.$ Since our markets have
a large number of investors, it is preferable instead to work with the limit
of these laws. We note from $(2)$ above that for each $n$, $\lambda_{N,n}\rightarrow
((m-1)n+1)$ as an increasing sequence in $N$.

\bigskip Let $p_{n}(t)$ $(\triangleq p_{t}(n))$ be the solution of the affine
Kolmogorov system of equations:
\begin{align*}
\frac{dp_{t}(0)}{dt}  &  =-p_{t}(0)\text{
\ \ \ \ \ \ \ \ \ \ \ \ \ \ \ \ \ \ \ \ \ \ \ \ \ \ \ \ \ \ \ }(3)\\
\frac{dp_{t}(n)}{dt}  &  =((m-1)(n-1)+1)p_{t}(n-1)-((m-1)n+1)p_{t}(n)\text{
\ ; }n\geq1\text{.}%
\end{align*}

Recall fron the first section that $\mu_{t}=%
{\displaystyle\sum\limits_{n\geq0}}
p_{n}(t)\frac{1}{\#_{m}(n)}%
{\displaystyle\sum\limits_{A_{n}\in\mathbb{A}_{n}}}
\mu^{\circ_{m}A_{n}}$.

\begin{proposition}
The sequence of \ laws $\mu_{t}^{\ast,N}$ converges to $\mu_{t}$ as $N$ increases.
\end{proposition}

\bigskip

\begin{proof}
By Kurtz [9], we have that $p_{N,n}(t)\rightarrow p_{n}(t)$ as $N$ increases.
But $(p_{n}(t))_{n\geq0}$ is a probability law, so for $\epsilon>0,$ there
exists $n(\epsilon)$ such that

$%
{\displaystyle\sum\limits_{n\geq n(\epsilon)}}
p_{n}(t)<\epsilon.$ Now let $N(\epsilon)$ be such that $N>N(\epsilon)$ implies
that $|p_{N,n}(t)-p_{n}(t)|<\frac{\epsilon}{n(\epsilon)}$ for $0\leq n\leq
n(\epsilon).$We then have for $C\in\mathcal{E}$ and $N>N(\epsilon)$%

\[
|\mu_{t}^{\ast,N}(C)-\mu_{t}(C)|\leq%
{\displaystyle\sum\limits_{n=0}^{n(\epsilon)}}
|p_{N,n}(t)-p_{n}(t)|+2\epsilon\leq3\epsilon
\]
since $\frac{1}{\#_{m}(n)}%
{\displaystyle\sum\limits_{A_{n}\in\mathbb{A}_{n}}}
\mu^{\circ_{m}A_{n}}(C)\leq1$ and ($p_{N,n}(t))_{n\geq0}$ are probability
laws. Our claim is proved.
\end{proof}

\begin{lemma}
$p_{t}(n)=\frac{\#_{m}(n)}{(m-1)^{n}n!}e^{-t}(1-e^{-(m-1)t})^{n}$
\end{lemma}

\bigskip

\begin{proof}
\bigskip We need to solve the affine Kolmogorov system of equations $(3)$.

Proceeding by induction we have:\bigskip%
\begin{align*}
\frac{dp_{t}(0)}{dt}  &  =e^{-t}\\
\frac{dp_{t}(n)}{dt}  &  =((m-1)(n-1)+1)e^{-(n(m-1)+1)t}%
{\displaystyle\int\limits_{0}^{t}}
e^{(n(m-1)+1)s}p_{s}(n-1)ds\text{ \ }%
\end{align*}
To prove the lemma it suffices to note that $\#_{m}(n)=\#_{m}%
(n-1)((n-1)(m-1)+1)$ and that $e^{(n(m-1)+1)s}e^{-s}(1-e^{-(m-1)s}%
)^{n-1}=e^{(m-1)s}(e^{(m-1)s}-1)^{n-1}$ is the derivative of $\frac{1}%
{(m-1)n}(e^{(m-1)s}-1)^{n}.$
\end{proof}

And this shows that the limit law of $P$ is indeed the extended\ Wild sum
which we have shown (in [1]) to be the solution of the ODE associated to the
interacting system.

\section{\bigskip Explicit formulas for other OTC market models}

In many applications, it is more convenient to work with more than one kernel
to describe the dynamics of the system. It is the case for instance in the
models of Duffie-G\^{a}rleanu-Pedersen [6] and their extensions in
B\'{e}langer-Giroux-Moisan [2] and in B\'{e}langer-Giroux-Ndoun\'{e} [3].

In the simplest such model on $E=\mathbb{\{(}l,n),(l,o),(h,n),(h,0)\}$ we have
the binary kernel we described at the beginning of section 3 and we have the
autonomous changes of liquidity of an investor. Let $\gamma_{u}$ and $\gamma_{d}$ resp. be the intensity of the up movements
(resp. down movements) in liquidity. We will first assume that these intensities are equal (we will remove this assumption at the end of the section) and we let
$\gamma=\gamma_{u}=\gamma_{d}.$ Then $q_{p}(t)=e^{-\gamma t}\frac{(\gamma
t)^{p}}{p!}$ is the probability of having $p$ autonomous movements up to time
$t.$ 
The 1-ary kernel can then be defined by
\begin{align*}
Q_{1}((l,n);C)  &  =\delta_{(h,n)}(C);Q_{1}((l,o);C)=\delta_{(h,o)}(C);\\
Q_{1}((h,n);C)  &  =\delta_{(l,n)}(C);Q_{1}((h,o);C)=\delta_{(l,0)}(C);
\end{align*}

and%
\[
\nu^{\circ_{1}}(C)=\nu(l,n)\delta_{(h,n)}(C)+\nu(l,o)\delta_{(h,o)}%
(C)+\nu(h,n)\delta_{(l,n)}(C)+\nu(h,o)\delta_{(l,0)}(C)
\]
for $C\subset E.$

It is possible to modify the 1-ary kernel into a binary kernel with the
addition of a ``witness" investor who is completely unaffected by the change of
liquidity of the other investor. We then symmetrize that kernel and replace
the two binary kernels with a convex combination of the two kernels to be left
with only one kernel as in the situation we dealt with in the preceding
sections. So all the formalism developed so far is still valid. The drawbacks
of this approach though are that the symmetrization operation gives us a
slightly different dynamics, and more importantly, that we lose the explicit
formulas for the extended Wild sums. The objective of this section is to show
how we can retain them.

Let $\mathbb{K}_{n}^{p}$ denote the set of all arrangements of $p$
undistinguishable objects in $n$ boxes, where a box may contain arbitrarily
many objects. Then $|\mathbb{K}_{n}^{p}|=\binom{n+p-1}{n-1}.$ Let
$\mathbb{A}_{n}$ denote, as before, the set of all random trees with $n$
$m-$ary interactions (the investors meetings). If $A_{n}\in\mathbb{A}_{n}$
then $A_{n}$ has $(m-1)n+1$ leaves which in turn gives $mn+1$ branches.

Let $p$ denote the number of $1-$ary interactions (i.e. the number of
autonomous changes of position). For $\sigma\in\mathbb{K}_{mn+1}^{p},$ let
$A_{n}^{\sigma}$ denote the tree obtained by placing the 1-ary interactions on
each branch of the tree according to the arrangement $\sigma.$ Let
$\widetilde{\mathbb{A}}_{n,p}$ denote the set of trees with $n$ $m$-ary
interactions and $p$ 1-ary ones. Then $\rho:\mathbb{K}_{mn+1}^{p}%
\times\mathbb{A}_{n}\rightarrow\widetilde{\mathbb{A}}_{n,p}:(\sigma
,A_{n})\mapsto A_{n}^{\sigma}$ defines a bijection. Please see figure 2 for simple examples of trees in $\widetilde{\mathbb{A}}_{2,1}.$ Moreover, if we call $\{\sigma^{i}\}_{i=1}^{7}$ the 7 configurations of figure 3, 
then $(\nu^{\circ_{3}A_{2}})^{\circ_{1}}= \sum_{i=1}^{7}\nu^{\circ_{3}A_{2}^{\sigma^{i}}}.$

The set $\widetilde{\mathbb{A}}_{n,p}$ of $(m,1)$-ary ordered trees
with $n$ $m$-ary internal nodes and $p$ 1-ary nodes has the cardinality equal to
$\#_{m}(n)\binom{mn+p}{mn}.$ Then $\widetilde{\mathbb{A}}_{n,p}$ constitutes a
set random trees if we assume that every $(m,1)$-ary tree in $\widetilde
{\mathbb{A}}_{n,p}$ is equally likely, namely with probability $\frac{1}%
{\#_{m}(n)\binom{mn+p}{mn}}.$


\begin{figure}
[ptb]
\begin{center}
\includegraphics[
natheight=1.100000in,
natwidth=5.074700in,
height=1.1338in,
width=5.1301in
]%
{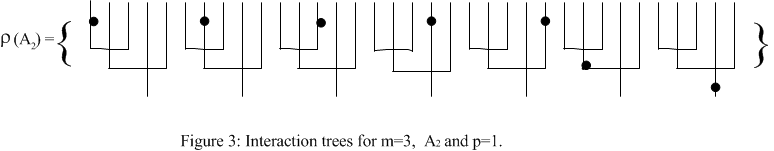}%
\end{center}
\end{figure}


\begin{proposition}
\bigskip The probability measure%
\[
\nu_{t}=%
{\displaystyle\sum\limits_{n\geq0}}
p_{n}(t)\frac{1}{\#_{m}(n)}%
{\displaystyle\sum\limits_{A_{n}\in\mathbb{A}_{n}}}
\left[
{\displaystyle\sum\limits_{p=0}^{\infty}}
\frac{q_{p}(t)}{\binom{mn+p}{mn}}%
{\displaystyle\sum\limits_{\sigma\in\mathbb{K}_{mn+1}^{p}}}
\nu^{\circ_{m}A_{n}^{\sigma}}\right]  \text{\ \ \ }%
\]
is the solution of the Cauchy problem:
\[
\frac{d\nu_{t}}{dt}=\lambda(\nu_{t}^{\circ_{m}}-\nu
_{t})+\gamma\;(\nu_{t}^{\circ_{1}}-\nu_{t});\nu_{0}%
=\nu.
\]

\end{proposition}

\begin{proof}

Let $\mu_t$ be the solution (1) above. Since $|\nu_t|\leq |\mu_t|$ and $\mu_t$ is uniformly summable, the convex
sum $\nu_t$ can be differentiated term by term to obtain:\\

$\frac{d\nu_{t}}{dt}=-\lambda \nu_{t}+$

\[
\lambda e^{-m\lambda t}{\displaystyle\sum\limits_{n\geq1}}
(1-e^{-(m-1)\lambda t})^{n-1}\frac{1}{(m-1)^{n-1}(n-1)!}%
{\displaystyle\sum\limits_{A_{n}\in\mathbb{A}_{n}}}
\left[
{\displaystyle\sum\limits_{p=0}^{\infty}}
\frac{q_{p}(t)}{\binom{mn+p}{mn}}%
{\displaystyle\sum\limits_{\sigma\in\mathbb{K}_{mn+1}^{p}}}
\nu^{\circ_{m}A_{n}^{\sigma}}\right]\text{\ \ \ }%
\]
\[
-\gamma \nu_{t}+%
{\displaystyle\sum\limits_{n\geq0}}
p_{n}(t)\frac{1}{\#_{m}(n)}%
{\displaystyle\sum\limits_{A_{n}\in\mathbb{A}_{n}}}
\left[
{\displaystyle\sum\limits_{p=1}^{\infty}}
\frac{e^{-\gamma t}(\gamma t)^{p-1}}{(p-1)! \binom{mn+p}{mn}}%
{\displaystyle\sum\limits_{\sigma\in\mathbb{K}_{mn+1}^{p}}}
\nu^{\circ_{m}A_{n}^{\sigma}}\right]  \text{\ \ \ }%
\]

Now using the combinatorial results of the proof of theorem 1 in B\'{e}langer-Giroux
[1], we have that $\nu_{t}^{\circ_{m}}$ is equal to the expression

\[
\lambda e^{-m\lambda t}{\displaystyle\sum\limits_{n\geq1}}
(1-e^{-(m-1)\lambda t})^{n-1}\frac{1}{(m-1)^{n-1}(n-1)!}%
{\displaystyle\sum\limits_{A_{n}\in\mathbb{A}_{n}}}
\left[
{\displaystyle\sum\limits_{p=0}^{\infty}}
\frac{q_{p}(t)}{\binom{mn+p}{mn}}%
{\displaystyle\sum\limits_{\sigma\in\mathbb{K}_{mn+1}^{p}}}
\nu^{\circ_{m}A_{n}^{\sigma}}\right]\text{\ \ \ }%
.\]
Otherwise, the expression
\[
%
{\displaystyle\sum\limits_{n\geq0}}
p_{n}(t)\frac{1}{\#_{m}(n)}%
{\displaystyle\sum\limits_{A_{n}\in\mathbb{A}_{n}}}
\left[
{\displaystyle\sum\limits_{p=1}^{\infty}}
\frac{e^{-\gamma t}(\gamma t)^{p-1}}{(p-1)! \binom{mn+p}{mn}}%
{\displaystyle\sum\limits_{\sigma\in\mathbb{K}_{mn+1}^{p}}}
\nu^{\circ_{m}A_{n}^{\sigma}}\right]  \text{\ \ \ }%
\] is equal to the quantity

\[
%
{\displaystyle\sum\limits_{n\geq0}}
p_{n}(t)\frac{1}{\#_{m}(n)}%
{\displaystyle\sum\limits_{A_{n}\in\mathbb{A}_{n}}}
\left[
{\displaystyle\sum\limits_{p=0}^{\infty}}
\frac{e^{-\gamma t}(\gamma t)^{p}}{p! \binom{mn+p+1}{mn}}%
{\displaystyle\sum\limits_{\sigma\in\mathbb{K}_{mn+1}^{p+1}}}
\nu^{\circ_{m}A_{n}^{\sigma}}\right].  \text{\ \ \ }%
\]

But
$
{\displaystyle\sum\limits_{\sigma\in\mathbb{K}_{mn+1}^{p+1}}}
\nu^{\circ_{m}A_{n}^{\sigma}}  \text{\ \ \ }%
=
{\displaystyle\sum\limits_{\sigma'\in\mathbb{K}_{mn+1}^{p}}}
(\nu^{\circ_{m}A_{n}^{\sigma'}})^{\circ_{1}}  \text{\ \ \ }.%
$
And then

\[
%
{\displaystyle\sum\limits_{n\geq0}}
p_{n}(t)\frac{1}{\#_{m}(n)}%
{\displaystyle\sum\limits_{A_{n}\in\mathbb{A}_{n}}}
\left[
{\displaystyle\sum\limits_{p=0}^{\infty}}
\frac{e^{-\gamma t}(\gamma t)^{p}}{p! \binom{mn+p+1}{mn}}%
{\displaystyle\sum\limits_{\sigma'\in\mathbb{K}_{mn+1}^{p}}}
(\nu^{\circ_{m}A_{n}^{\sigma'}})^{\circ_{1}}\right]=\nu_{t}^{\circ_{1}}.  \text{\ \ \ }%
\]

Hence we get that $\frac{d\nu_{t}}{dt}$ has the desired form and the proof of the proposition
is now complete.

\end{proof}

We note that if $p=0,$ that is, if no investor changes its liquidity position, the
above solution does indeed  become the solution $(1).$


\begin{remark}
In the specific context of the DGP model we have a binary kernel which simplifies the first part of the formula. But without the assumption $\gamma_u=\gamma_d$, we have to consider the up movements and the down movements separately, and this makes for a more complicated second part of the formula.  Let $\gamma=\gamma_u+\gamma_d$, 
the solution becomes

\[
\nu_{t}=%
{\displaystyle\sum\limits_{n\geq0}}
\frac{e^{-t}(1-e^{-t})}{n!}%
{\displaystyle\sum\limits_{A_{n}\in\mathbb{A}_{n}}}
\left[
{\displaystyle\sum\limits_{p=0}^{\infty}}
\frac{e^{-\gamma t}}{p!\binom{p}{k}\binom{mn+p+1}{mn}}%
\left[ {\displaystyle\sum\limits_{k=0}^{p}}(\gamma_{u})^{k}(\gamma_{d})^{p-k} {\displaystyle\sum\limits_{\underset{\sigma_{d}\in\mathbb{K}_{mn+1}^{p-k}}{\sigma_{u}\in\mathbb{K}_{mn+1}^{k}}}}
\nu^{\underset{m}\circ A_{n}^{\sigma_{u}\cup \sigma_{d}}}\right] \right]  \text{\ \ \ }%
\]

where $\sigma_{u}$ (resp $\sigma_{d}$) denotes the arrangements of up movements (resp. down movements) on the branches of the tree and $\sigma_{u}\cup \sigma_{d}$ is the arrangement obtained from both arrangements of up and down movements.

\end{remark}

\begin{remark}
We can obtain similar explicit formulas for OTC models where the interactions involve $m>2$ investors. In the information percolation model of Duffie-Malamud-Manso[7], for instance, the state space, $E=\mathbb{N}$ represents the potential levels of information acquired by an investor through meetings with other investors. The $m$-ary interaction is the perfect sharing of information which means that each investor in the meeting comes out with the sum of the information levels of all participating investors. The unary kernel is a regression force which replaces an investor of level $n$ say, by an investor with level $\pi(n)$ sampled from a given distribution $\pi$ on $\mathbb{N}$.
\end{remark}

\bigskip
\textbf{Acknowledgement:} This research is supported in part by a team grant
from Fonds de Recherche du Qu\'{e}bec - Nature et Technologies (FRQNT grant no. 180362).

\section{References}

\begin{enumerate}
\item B\'{e}langer, A. and Giroux, G., (2013), Some new results on information
percolation, Stochastic Systems, vol. 3, 1-10.

\item B\'{e}langer, A., Giroux, G. and Moisan-Poisson, M. (2013),
Over-the-Counter Market Models with Several Assets, Arxiv 1308.2957v1.

\item B\'{e}langer, A., Giroux, G. and Ndoun\'{e}, N. (2014), Existence of
Steady States for Over-the-Counter Market Models with Several Assets, Arxiv 1126039.

\item Carlen, E., Carvalho, M.C. and Gabetta, E. (2005), On the relation
between rates of relaxation and convergence of Wild sums for solutions of the
Kac equation, Journal of Functional Analysis, 220, no. 2, 362-387.

\item Duffie, D. (2012). Dark Markets: Asset Pricing and Information
Percolation in Over-the-Counter Markets. Princeton Lecture Series.

\item Duffie, D., G\^{a}rleanu, N. and Pedersen, L.H. (2005), Over-the-counter
markets, Econometrica 73, 1815-1847.

\item Duffie, D., Malamud, S. and Manso, G. (2009), Information percolation with equilibrium search dynamics, Econometrica 77, 1513-1574.

\item Kac, M. (1956). Foundations of kinetic theory. Proceedings of the Third
Berkeley Symposium on Mathematical Statistics and Probability, 1954--1955,
vol. \textbf{III}, pp. 171--197. University of California Press, Berkeley and
Los Angeles.

\item Kurtz, T. G.(1969). A Note on sequences of continuous parameter Markov
chains. Ann. Math. Statist., \textbf{40}, 1078-1082.

\item Lefebvre, M. (2006). Applied Stochastic Processes. Springer.

\item Tanaka, H. (1969), Propagation of Chaos for Certain Markov Processes of
Jump Type with Nonlinear Generators II, Pro. Japan Acad. 45, 598-600

\item Tanaka,\ S. (1968), An extension of Wild's Sum for Solving Certain
Non-Llinear Equation of Measures, Proc. Japan Acad. 44, 884-889.

\item Pareschi, L., Caflischt, R.E. and Wennberg, B. (1999), An Implicit Monte
Carlo Method for Rarefied Gas Dynamics, J.Comput.Physics. 154, 90-116.

\item Trazzi, S., Pareschi, L. and Wennberg, B. (2009), Adaptive and Recursive
Time Relaxed Monte Carlo Methods for Rarefied Gas Dynamics, SIAM J.Sci.Comput.
31(2), 1379-1398.

\item Wild, E (1951). On the Boltzmann equation in the kinetic theory of
gases. Proc. Cambridge Phil. Soc. \textbf{47}, 602-609.\bigskip

$%
\begin{array}
[c]{l}%
\text{Alain B\'{e}langer, Gaston Giroux, Ndoun\'{e} Ndoun\'{e}}\\
\text{D\'{e}partement de finance,}\\
\text{Universit\'{e} de Sherbrooke,}\\
\text{2 500 boul. de l'Universit\'{e},}\\
\text{Sherbrooke, Canada, J1K 2R1}\\
\text{E-addresses: alain.a.belanger@usherbrooke.ca}\\
\text{gasgiroux@hotmail.com, ndoune.ndoune@usherbrooke.ca}%
\end{array}
$
\end{enumerate}

\end{document}